\documentclass[runningheads]{llncs}
\usepackage{soul}
\newcommand*\samethanks[1][\value{footnote}]{\footnotemark[#1]}
\usepackage[utf8]{inputenc}
\usepackage{pifont}
\usepackage[x11names,dvipsnames,table]{xcolor} 
\usepackage{mathtools}
\usepackage{amsfonts,amssymb}
\usepackage{caption,lineno}
\usepackage{mathtools}

\usepackage{color}
\usepackage{amssymb}
\usepackage{graphicx}
\usepackage{lscape}
\usepackage{rotating}
\usepackage{siunitx}
\usepackage[most]{tcolorbox}

\usepackage{ifthen}
\usepackage[noend]{algpseudocode}

\usepackage{algorithm}
\newcommand{\nph}{\mathsf{NP\text{-}Hard}}
\newcommand{\npc}{\mathsf{NP\text{-}Complete}}
\newcommand{\threesat}{\mathsf{3\text{-}SAT}}
\newcommand{\mcsc}{\mathsf{MCSC}}
\newcommand{\smcsc}{\mathsf{S\text{-}MCSC}}
\newcommand{\lmcsc}{\mathsf{L\text{-}MCSC}}
\newcommand{\apx}{\mathsf{APX\text{-}Hard}}

\newcommand{\remove}[1]{}

\usepackage{simplemargins}
\setallmargins{1in}

\begin{document}

\title{Minimum Color Spanning Circle \\ of Imprecise Points\thanks{A preliminary 
version of this paper  appeared in the Proceedings of the 27th International Conference on Computing and Combinatorics (COCOON 2021).}} 


\author{Ankush Acharyya\inst{1}\thanks{The work was partially done while affiliated with Institute of Computer Science of the Czech Academy of Sciences, with institutional support RVO:67985807.}\and
Ramesh K. Jallu\inst{2}\samethanks[2] \and
Vahideh Keikha\inst{3} \and \\ Maarten L\"{o}ffler\inst{4} \and Maria Saumell \inst{3,5}}

\institute{Dept. of Mathematics and Computing, Indian Institute of Technology (ISM) Dhanbad, India\\ \email{ankush@iitism.ac.in} \\
\and Dept. of Computer Science and Engineering, Indian Institute of Information Technology Raichur, India\\\email{jallu@iiitr.ac.in}\\
\and 
The Czech Academy of Sciences, Institute of Computer Science, Czech Republic\\ \email{keikha@cs.cas.cz}\\
\and
Dept. of Information and Computing Sciences, Utrecht University, the Netherlands\\
\email{M.Loffler@uu.nl}\\
 \and
Department of Theoretical Computer Science, Faculty of Information Technology, Czech Technical University in Prague, Czech Republic\\
\email{maria.saumell@fit.cvut.cz}
}


\authorrunning{A. Acharyya, R.\,K. Jallu, V. Keikha, M. L\"{o}ffler, M. Saumell}

\maketitle

\begin{abstract}
Let $\cal R$ be a set of $n$ colored imprecise points, where each point is colored by one of $k$ colors. Each imprecise point is specified by a unit disk in which the point lies. We study the problem of computing the smallest and the largest possible minimum color spanning circle, among all possible choices of points inside their corresponding disks. We present an $O(nk\log n)$ time algorithm to compute a smallest minimum color spanning circle. Regarding the largest minimum color spanning circle, we show that the problem is $\nph$ and present a $\frac{1}{3}$-factor approximation algorithm. We improve the approximation factor to $\frac{1}{2}$ for the case where no two disks of distinct color intersect.

\keywords{Color spanning circle \and Imprecise points \and Algorithms \and Computational complexity}
\end{abstract}

\section{Introduction}
Recognition of color spanning objects of optimum size, in the classical (precise) setting, is a well-studied problem in the literature~\cite{acharyya2018minimum,de2005tsp,das2009smallest,huttenlocher1993upper}. The motivation of color spanning problems stems from facility location problems. Here facilities of type $i\in \{1,2,\ldots,k\}$ are modeled as points with
color code $i$, and the objective is to identify the location of a desired geometric shape containing at
least one facility of each type such that the desired measure parameter (width, perimeter, area, etc.) is optimized. Other applications of color spanning objects can be found in disk-storage management systems~\cite{ConsuegraN13} and central-transportation systems~\cite{manzini2008design}.

The  simplest type of two-dimensional problem considered in this setup is the {\em minimum color spanning circle} ($\mcsc$) problem, defined as follows.
Given a colored point set $P$ in the plane, such that each point in $P$ is colored with one of $k$ possible colors, compute a circle of minimum radius that contains at least one point of each color (see Fig.~\ref{fig:defn}a). The $\mcsc$ problem is well-understood: As observed  by Abellanas et al.~\cite{abellanas2001smallest}, a minimum color spanning circle can be computed in $O(nk\log n)$ time using results on the upper envelope of Voronoi surfaces obtained by Huttenlocher et al.~\cite{huttenlocher1993upper}. 
 
\begin{figure*}[t]
    \centering
    \includegraphics[width=0.9\linewidth]{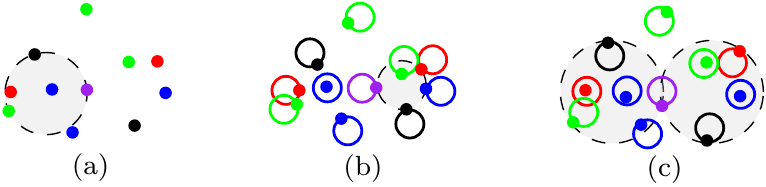}
    \caption{(a) $\mcsc$ for a precise colored point set. (b,c) $\smcsc$ and $\lmcsc$ for an imprecise colored point set. The representative for each disk is marked as a point of the corresponding color. 
    In (c), we display two $\lmcsc$ of the same radius.}
    \label{fig:defn}
\end{figure*}

A limitation of these studies on color spanning objects is that, in many real-life situations, the locations of the points are subject to errors and their exact coordinates are unknown. This is an issue that might arise every time that we try to apply a geometric algorithm to real-world data. For this reason, in recent years there have been many studies aimed at describing how imprecision in the input data might affect some of the most fundamental geometric descriptors, such as the convex hull or the smallest bounding box. In these studies, the input is an {\em imprecise} or {\em uncertain} point set: the location of every point is not uniquely determined, but given by a set of possible locations called its {\em region}~\cite{loffler2010largest,salesin1989epsilon}. One of the fundamental questions to be studied in this setting is to compute the extreme values of the descriptors, that is, the smallest and biggest, e.g., convex hull of the point set, for all possible locations of every point in its region. In this paper, we consider this problem for the minimum color spanning circle.

\subsection{Related work} 

Uncertainty in data is paramount in contemporary geometric computations. In the literature, different variations have been considered  where the regions are modelled as simple geometric objects such as line segments, disks or squares \cite{loffler2010largest,salesin1989epsilon}. 
Computing the smallest circle  intersecting a set of disks or convex regions of total complexity $n$ is called the {\em intersection radius problem}, and can be solved in $O(n)$ time~\cite{jadhav1996optimal}. 
  Robert and Toussaint~\cite{robert1990computational} studied the problem of computing the smallest width corridor intersecting a set of convex regions (disks and line segments) and proposed two $O(n\log n)$ time algorithms, where $n$ is the number of convex regions. L\"{o}ffler and van Kreveld~\cite{loffler2010largest} considered the problem of computing the smallest and largest possible axis-parallel bounding box and circle of a set of regions modelled as circles or squares. Their proposed algorithms have running times ranging from $O(n)$ to $O(n\log n)$. 
  For a set of squares or line segments, computing a placement of points in the regions that maximizes or minimizes the area or the perimeter of the convex hull is studied by the same authors~\cite{LofflerK10}. Some variants are shown to be $\nph$, and the polynomial time algorithms have running times ranging from $O(n \log n)$ to $O(n^{13})$. If the input is a set of disks, a $(1+\varepsilon)$-approximation algorithm for this problem is given also by van Kreveld and L\"{o}ffler~\cite{KreveldL08}.
  Other problems that have been studied in the region-based model are computing a placement to maximize or minimize the diameter on a set of squares or disks~\cite{KeikhaLM20,loffler2010largest}, or the area of the largest or smallest triangle on a set of line segments~\cite{KeikhaLM21}.

  We note that other model formulations have also been proposed for dealing with inaccuracies in geometric problems. These are epsilon geometry~\cite{salesin1989epsilon}, probabilistic models~\cite{cormode2008approximation,suri2013most}, the aggregated uncertainty model~\cite{keikha2021clustering}, and the domain based models~\cite{edalat2001convex}.  

Regarding other color spanning objects 
in the precise setting, efficient algorithms are known for computing smallest color spanning squares~\cite{abellanas2001smallest}, strips and  rectangles~\cite{das2009smallest}, 2-intervals~\cite{jiang2014shortest}, equilateral triangles of fixed orientation~\cite{hasheminejad2015computing}, and axis-parallel squares~\cite{khanteimouri2013computing}. 
Acharyya et al.~\cite{acharyya2018minimum} propose efficient algorithms to compute the narrowest color spanning annulus for circles, axis-parallel squares, rectangles, and equilateral triangles of fixed orientation. The  minimum diameter color spanning set problem has also been studied~\cite{FleischerX11,JuFLZD13,ZhangCMTK09}. Its general version is known to be $\nph$ in the $L_p$ metric, for $1<p<\infty$, while in the $L_1$ and $L_\infty$ metrics the problem can be solved in polynomial time~\cite{FleischerX11}. 

Colored variations of other geometric problems have also been studied in the context of imprecise points~\cite{DaescuJL10,dror2008combinatorial,pop2019generalized}. Given a set of colored clusters, the problem of computing the minimum-weight color spanning tree ({\it generalized MST problem}) is $\apx$~\cite{dror2008combinatorial}. Even when each cluster contains exactly two points the problem remains $\nph$~\cite{fraser2013algorithms}. The problem admits a $2\delta$-approximation, where $\delta$ is the maximum size of the cluster for any imprecise vertex of the MST~\cite{pop2019generalized}. In the generalized TSP problem (GTSP), the imprecision is defined by neighborhoods (which are either continuous or discrete) and the goal is to find the shortest tour that visits all neighborhoods. It is known that GTSP with neighborhoods defined by subsets of cardinality two is inapproximable~\cite{dror2008combinatorial}. 


Finally, we would like to notice that the problem of computing the largest minimum color spanning circle (see the formal description below) is closely related to the dispersion problem in unit disks, where for a given set of $n$ unit disks the goal is to select $n$ points, one from each disk, such that the minimum pairwise distance among the selected points is maximized.
This problem was introduced by Fiala et al.~\cite{fiala2005systems}, and the authors  proved that it is $\nph$. It is also known that the problem is APX-hard~\cite{dumitrescu2012dispersion}. Constant factor approximation algorithms for this problem are given by Cabello~\cite{cabello2007approximation}, and by Dumitrescu and Jiang~\cite{dumitrescu2012dispersion}. 

\subsection{Problem definition and results}

  In this work, we are given a set ${\cal R}=\{R_1,R_2,\ldots,R_n\}$ of $n$ unit disks of {\em diameter~$1$} in the plane, where each disk is colored with one of $k$ possible colors. A colored point set $P$ is a {\em realization} of $\cal R$ if there exists a color-preserving bijection between $P$ and $\cal R$ such that each point in $P$ is contained in the corresponding disk in $\cal R$. Each realization of $\cal R$ gives a $\mcsc$ of certain radius. We are interested in finding realizations of $\cal R$ such that the corresponding $\mcsc$ has the smallest ($\smcsc$) and largest ($\lmcsc$) possible radius (see Fig.~\ref{fig:defn}b and c). 

We present the following results:
\begin{itemize}
    \item The $\smcsc$ problem can be solved in $O(nk\log n)$ time.
    \item The $\lmcsc$ problem is $\nph$.
    \item A $\frac{1}{3}$-factor approximation to the $\lmcsc$ problem can be computed in $O(nk\log n)$ time. When no two distinct color disks intersect, the approximation factor becomes $\frac{1}{2}$.
\end{itemize}

To the best of our knowledge there is no prior result on the minimum color spanning circle problem for imprecise point sets.

\section{The smallest $\mcsc$ ($\smcsc$) problem}\label{sec:smallest}

Given a set $\cal R$ of $n$ imprecise points modeled as unit disks, we present an algorithm that finds a $\smcsc$, denoted by $C_{opt}$, and the realization of $\cal R$ achieving it. Notice that the $\smcsc$ problem is equivalent  to  finding  a  smallest  circle  that  intersects  all  color  regions. 

Let ${\cal C}=\{c_1,\ldots,c_n\}$ be the set of center points of the disks in $\cal R$. Let $C_c$ be a $\mcsc$ of the colored set $\cal C$, and let $r_c$ be its radius. Finally, let $r_{opt}$ be the radius of~$C_{opt}$. The following relation holds:

\begin{lemma}\label{lem:equi_radii}
 If $r_c > \frac{1}{2}$, then $r_{opt} = r_c  - \frac{1}{2}$.
\end{lemma}
\begin{proof}

Consider a circle $C_c'$ concentric with $C_c$ with radius $r'_c = r_c - \frac{1}{2}$ (see Fig.~\ref{fig:lem1}). For every disk $R_i$ such that $c_i$ is contained in $C_c$, we have that $C'_c$ contains $c_i$ or the intersection between the boundary of $R_i$ and the segment connecting $c_i$ with the center of $C_c$. Thus, $C'_c$ contains at least one point of each color and $r_{opt}\leq r'_c$.

 \begin{figure}[t]
  \centering
  \includegraphics{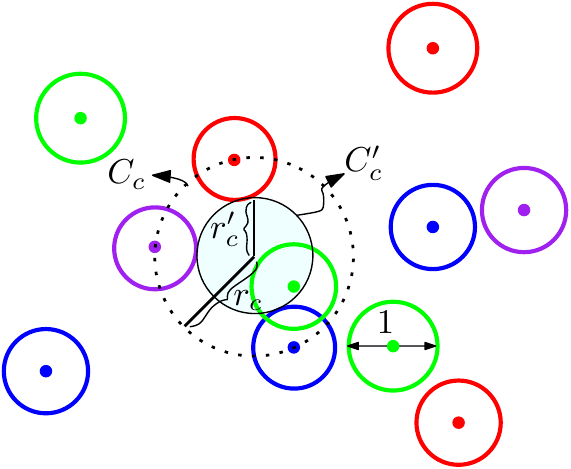}
  \caption{Illustration of Lemma~\ref{lem:equi_radii}. The dotted circle $C_c$ is the $\mcsc$ of the disk centers and the circle $C'_c$ is obtained by decreasing $C_c$'s radius by $\frac{1}{2}$.}\label{fig:lem1}
 \end{figure}

If $r_{opt} < r'_c$, we would get a feasible solution for the $\mcsc$ problem of $\cal C$ by increasing the radius of $C_{opt}$ by $\frac{1}{2}$. Since such a solution would have radius $r_{opt}+\frac{1}{2} < r'_c +\frac{1}{2}= r_c$, we would get a contradiction with the fact that $r_c$ is the radius of any $\mcsc$ of $\cal C$.  \hfill $\qed$
\end{proof}

Using Lemma~\ref{lem:equi_radii}, we compute $C_{opt}$ as described in Algorithm \ref{algo:smcsc}:\\

\begin{tcolorbox}[rightrule=2pt,toprule=0.5pt, bottomrule=2pt,leftrule=0.5pt,arc=8pt]
\captionof{algorithm}{Algorithm for the $\smcsc$ problem} \label{algo:smcsc}
\hrulefill\vskip -2mm
\begin{algorithmic}[1]
\Require {\it A set ${\cal R}$ of $n$ unit disks}
\Ensure {\it A $\smcsc$ of $\cal R$ with radius $r_{opt}$}
\State compute $C_c$;
\If {$r_c> 1/2$}
    \State $C_{opt}$ is a circle concentric with $C_c$, $r_{opt}\gets r_c-\frac{1}{2}$;
\Else
    \State $C_{opt}$ is a circle concentric with $C_c$, $r_{opt}\gets 0$; \Comment{$C_{opt}$ is a point}
\EndIf
\State \Return $C_{opt}$
\end{algorithmic}
\end{tcolorbox}


It only remains to prove that Algorithm \ref{algo:smcsc} is correct and efficient:

\begin{theorem}\label{thm:smcsc}
  A smallest minimum color spanning circle of $\cal R$ can be computed in $O(nk\log n)$ time.
\end{theorem}

\begin{proof}
In Algorithm \ref{algo:smcsc}, we first compute a $\mcsc$ $C_c$ of $\cal C$ in $O(nk\log n)$ time using the technique proposed by Huttenlocher et al.~\cite{huttenlocher1993upper}. If $r_c > \frac{1}{2}$, we shrink the radius of $C_c$ by $\frac{1}{2}$ and return this circle as a solution to the $\smcsc$ problem. The optimality of the solution follows from Lemma~\ref{lem:equi_radii}. If $r_c\leq \frac{1}{2}$, we have the following property: If $c_i$ is contained in $C_c$, then the center of $C_c$ in contained in~$R_i$. Hence, the center of $C_c$ lies in the intersection of $k$ distinct colored disks. In consequence, a circle of radius zero concentric with $C_c$ is a solution to the $\smcsc$ problem. \hfill $\qed$ 
\end{proof}

\section{The largest $\mcsc$ ($\lmcsc$) problem}


In this section, we consider the $\lmcsc$ problem, where for the given set $\cal R$ the goal is to find a realization such that any $\mcsc$ is as large as possible. We show that the problem is $\nph$, already for $k=2$, using a reduction from  planar $\threesat$ \cite{lichtenstein1982planar}.
Our reduction is inspired by those described by Fiala et al.~\cite{fiala2005systems}, and by Knauer et al.~\cite{KnauerLSW11}. 

Given a planar $\threesat$ instance, we construct a set of colored unit disks with the following property: There exists a realization such that any $\mcsc$ has diameter $\delta$ if and only if the $\threesat$ instance is satisfiable, where $\delta=\frac{9}{8}$.
We use disks of colors red and blue. Thus, a point set having any $\mcsc$ of diameter at least $\delta$ is equivalent to saying that there is no red-blue pair of points at distance less than $\delta$. We denote the family of realizations with this property by $\mathcal{P}^\delta$. With some abuse of notation, sometimes we also use the terminology $\mathcal{P}^\delta$ to characterize realizations of (only) a subset of the disks.

We next describe our construction.

\begin{figure}[t]
  \centerline{ \includegraphics{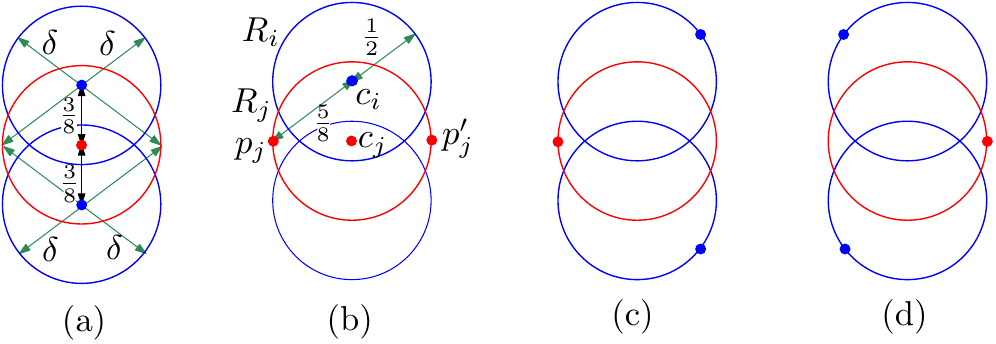}}
  \caption{(a) A stack of disks; (b) the distance from the left endpoint of the red disk to its farthest point in the top blue disk is $\frac 9 8$; (c,d) the two placements of points with red-blue distances equal to $\delta$.}\label{fig:stack}
\end{figure}

A \emph{stack of disks} is a set of three vertically aligned unit disks of alternating colors such that the center of the middle disk is at distance $\frac{3}{8}$ from the centers of the two other disks. See Fig.~\ref{fig:stack}a for an example of a blue-red-blue stack of disks.

\begin {lemma} \label{lem:stack}
  There exist two realizations in $\mathcal{P}^\delta$ of a stack of disks.
\end {lemma}

\begin{proof}
Let us denote by $R_i$ and $R_j$ the upper and middle disk of the stack, and let $c_i$ and $c_j$ denote their centers. Additionally, let the leftmost and rightmost points of $R_j$ be denoted by $p_j$ and $p'_j$ (see Fig.~\ref{fig:stack}b). Notice that, among all points which are on the upper half of $R_j$ or on its horizontal diameter, $p_j$ and $p'_j$ are the furthest to $c_i$.

Let $q_i,q_j$ be two points such that $q_i$ lies in $R_i$ and $q_j$ lies in $R_j$. Without loss of generality, we assume that $q_j$ lies on the upper half of $R_j$ or on its horizontal diameter (otherwise, we repeat the same arguments taking the lower disk of the stack instead of the upper one). Then we have


\begin{align*}
    d(q_j,q_i)&\leq d(q_j,c_i)+d(c_i,q_i) \leq d(p_j,c_i)+\frac{1}{2} &=\sqrt{(d(p_j,c_j))^2+(d(c_i,c_j))^2}+\frac{1}{2} \\ & =\frac{5}{8}+\frac{1}{2} =\frac{9}{8}.
\end{align*}

Notice that equality is only attained if $q_j,c_i$ and $q_i$ are aligned, and $q_j$ is equal to $p_j$ or $p'_j$. Thus, there exist only two realizations in $\mathcal{P}^\delta$, shown in Fig.~\ref{fig:stack}c and d. \hfill $\qed$
\end{proof}


\subsection{Variable Gadget} 

\begin{figure}[t]
  \centerline{ \includegraphics{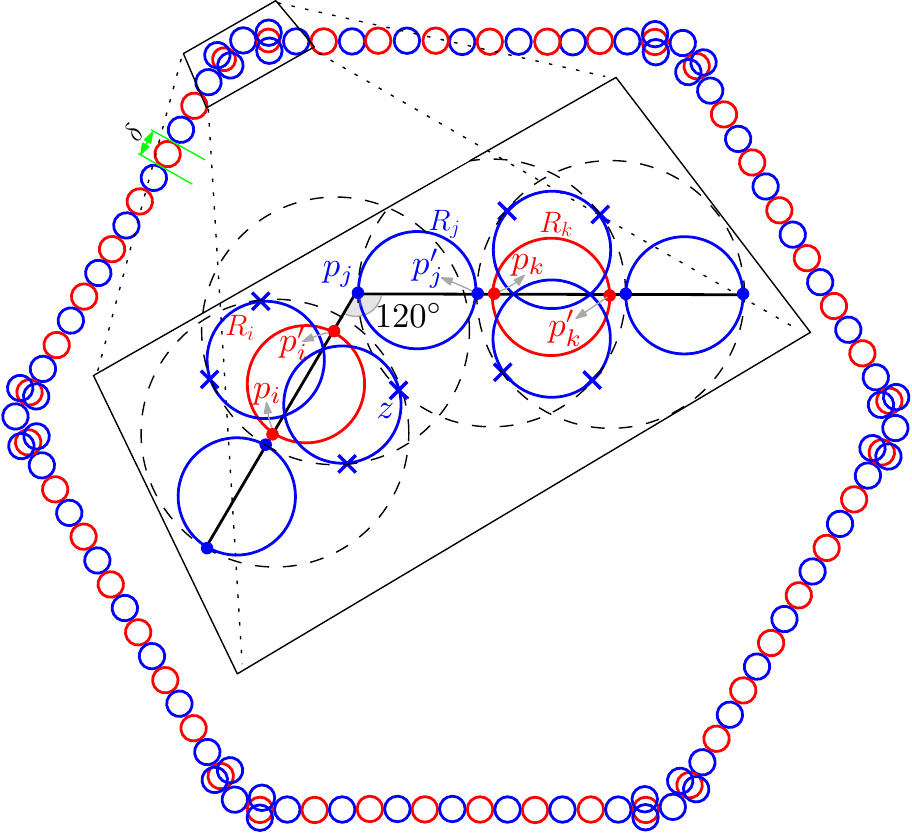}}
  \caption{A variable gadget with zoomed in view for the top-left corner. The dashed circles are centered at $p_k',p_k,p_i',p_i$ and have radius $\delta$.}\label{fig:var_gad}
\end{figure}

Our variable gadget (see Fig.~\ref{fig:var_gad}) is an alternating chain of red and blue disks, whose centers lie on a hexagon, together with some stacks of disks. The distance between the centers of two consecutive red and blue disks along the same edge of the hexagon is $\delta$. Each edge of the hexagon contains two stacks of disks placed near the endpoints, and every pair of consecutive edges is joined by a blue disk. In the following description, we say that $p_i$ and $p'_i$ are the \emph{leftmost} and \emph{rightmost} points of disk $R_i$ if they are its leftmost and rightmost points after the hexagon has rotated so that the edge containing the center of $R_i$ is horizontal and the center of the hexagon is below the edge. 

At the top-left corner of the variable gadgets, the disks are placed as follows (the other corners are constructed similarly). Let $R_i$ be the last disk in clockwise order along the top-left edge of the hexagon, and let $R_j$ and $R_k$ be the first and second disks along the top edge (see Fig.~\ref{fig:var_gad}).
The point $p_j$ lies at the top left corner of the hexagon, and the centers of $R_j$ and $R_k$ are at distance $\delta$. Regarding $R_i$, it is placed in such a way that the lower blue disk of its stack contains a point $z$ which is at distance $\delta$ from both $p_k$ and $p_i$ (see Fig.~\ref{fig:var_gad}). Notice that, if a realization in $\mathcal{P}^\delta$ chooses $p'_i$, the choice for $R_j$ is not unique; however, none of the points in $R_j$ at distance at least $\delta$ from $p'_i$ is at distance at least $\delta$ from $p_k$. Therefore, the choice of $p'_i$ forces the choice of $p'_k$, and clearly the choice of $p_k$ forces the choice of $p_i$. 


For a realization in $\mathcal{P}^\delta$ of a variable gadget, the following holds: By Lemma~\ref{lem:stack}, the stack containing $R_k$ is constrained to choose either $p_k$ or $p'_k$. Let us assume that it chooses $p'_k$. This choice propagates to the right through the chain of disks in the top edge. The red disk of the stack on the right of the edge also chooses its rightmost point, and this forces the red disk of the first stack of the top right edge to choose its rightmost point too. Therefore, the choice of $p'_k$ propagates through the whole hexagon. If $R_k$ chooses $p_k$, the same phenomenon occurs. We conclude:

\begin{lemma}\label{lem:var_two_states}
For any realization in $\mathcal{P}^\delta$ of a variable gadget, either all the unit disks centered at the edges of the hexagon, except for the ones intersecting corners of the hexagon, choose their rightmost point, or they all choose their leftmost point.
\end{lemma}


\subsection{Clause gadget} 

\begin{figure}[t]
\centering
  \includegraphics{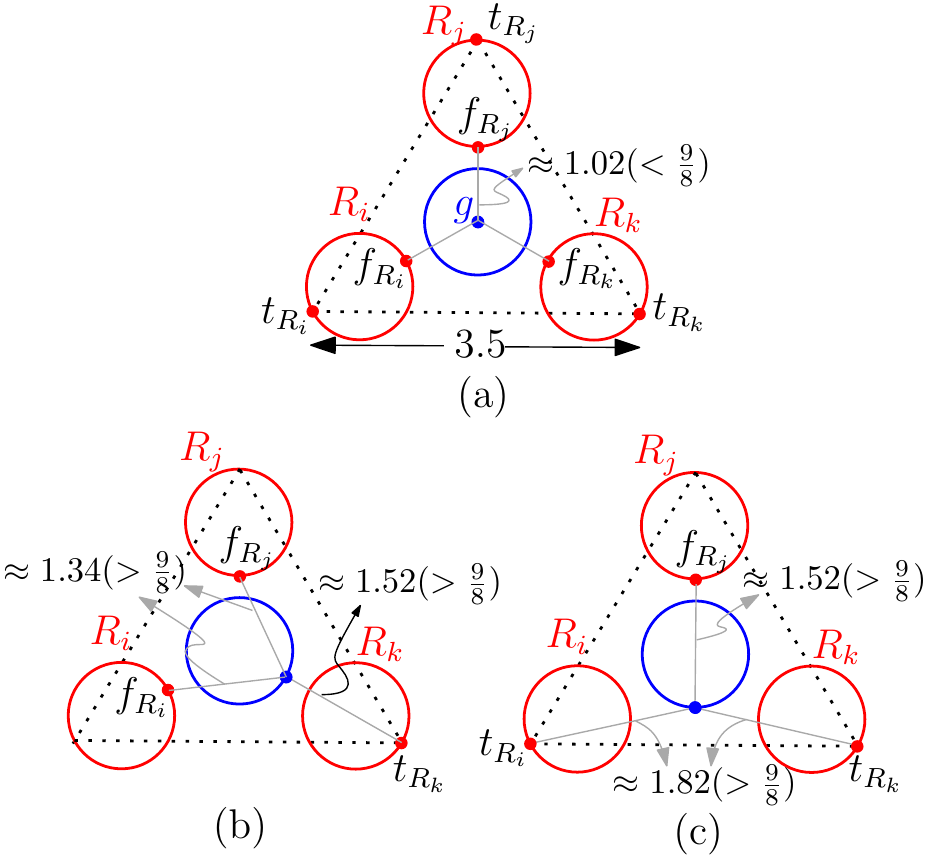}
   \caption{(a-c) A clause gadget with distinct truth assignments. In the placement in (a), the red-blue pairs are at distance smaller than $\delta$, while in (b) and (c) they are at distance greater than $\delta$.} \label{fig:clause-gadget}
\end{figure}

Clause gadgets are illustrated in Fig.~\ref{fig:clause-gadget}a-c. We consider an equilateral triangle of side length $\frac{7}{2}$ and center $g$, and we place one red disk at every corner of the triangle in such a way that the center of the disk is aligned with $g$ and its nearest corner of the triangle. Then we place a blue disk centered at $g$. Each red disk of a clause gadget is associated to one of the literals occurring in the clause, and is connected to the corresponding variable gadget via a connection gadget. Intuitively, to decide if there exists any realization in $\mathcal{P}^\delta$, each red disk $R_\tau$  of the clause gadget has essentially two relevant placements, called $t_{R_\tau}$ and $f_{R_\tau}$ (see Fig.~\ref{fig:clause-gadget}a). As we will see, when the associated literal is set to \emph{true}, we can choose the placement $t_{R_\tau}$, and when it is set to \emph{false}, we are forced to choose $f_{R_\tau}$. It is easy to see that any realization in $\mathcal{P}^\delta$ of the clause gadget does not choose $f_{R_\tau}$ for at least one of the disks $R_\tau$ (see Fig.~\ref{fig:clause-gadget}b and c).


\subsection{Connection gadget} 

\begin{figure}[t]
\centering
  \includegraphics{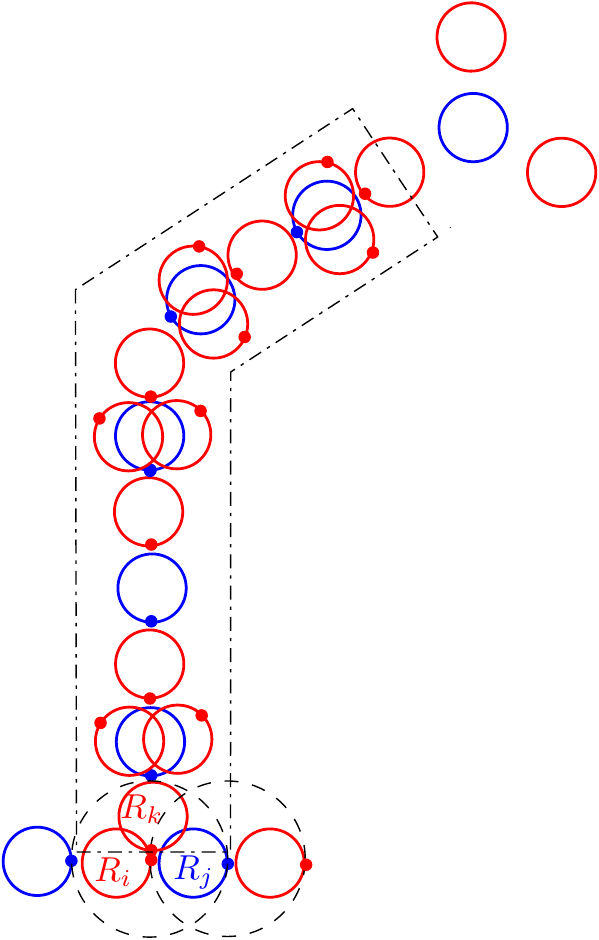}
   \caption{Connection gadget for a positive variable in a clause; the variable has truth value $T$.} \label{fig:connection-gadget}
\end{figure}

A variable gadget is connected to each of its corresponding clause gadgets with the help of a \emph{connection} gadget. A connection gadget consists of an alternating chain of red and blue disks together with some stacks of disks (see Fig.~\ref{fig:connection-gadget}). 

The chain in a connection gadget is formed by straight line portions of groups of disks equal to those along the edges of hexagons associated to  variable gadgets. These straight line portions are connected at bends of $120^\circ$. The placement of disks at each bend is similar to the placements at the corners of a variable gadget; in particular, stacks of disks are used around the bends. 

\begin{figure}[t]
  \centering
  \includegraphics{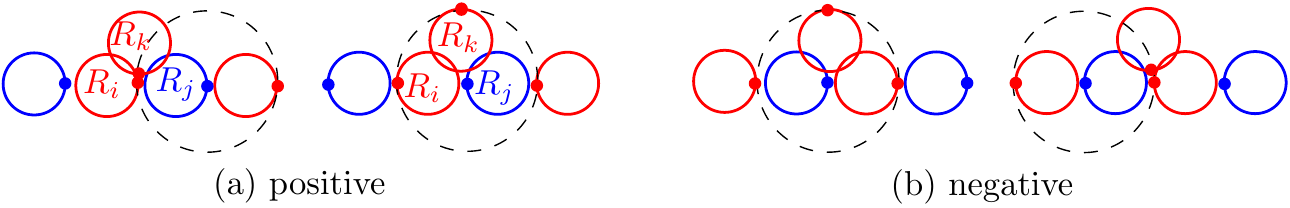}
  \caption{Point placements corresponding to (a) positive and (b) negative variable in a clause, at the intersection of a connection gadget and a variable gadget. 
     {\color{black} In both figures, the left subfigures correspond to the truth value $T$ for the variable, and the right subfigures to the truth value $F$.}
  }
 \label{fig:connecting}
 \end{figure}

The precise connection of the variable gadget to the clause gadget by a connection gadget depends on whether the variable in the clause is positive or negative. For a positive variable, the connection is established through a pair of a red disk $R_i$ and a blue disk $R_j$ which appear consecutively along an edge of the hexagon, and such that none of them intersects a corner of the hexagon, and $R_i$ comes before $R_j$ in clockwise order (see Fig.~\ref{fig:connecting}a). Let $R_k$ be the first red disk of the connection gadget. The top-most point of $R_k$ is at distance $\delta$ from $p_j$, and its bottom-most point is at distance $\delta$ from $p'_j$. For a negative variable, the connection is established through a pair of blue-red disks in the variable gadget. The placement of the first red disk of the connection gadget is analogous to the one in the red-blue configuration (see Fig.~\ref{fig:connecting}b). 

To ensure the desired propagation of the placement of points in the disks, stacks of disks are used next to the first disk $R_k$, and next to the red disk of the clause gadget associated to the literal (see Fig.~\ref{fig:connection-gadget}).

The truth value $T$ of a variable is associated with the choice of the rightmost points of the disks in the variable gadget. If the variable appears positive at a clause, this allows (for a realization in $\mathcal{P}^\delta$) the choice of the bottom-most point of $R_k$, and this propagates through the connection gadget and eventually allows the choice of the associated point $t_{R_\tau}$ in the clause gadget (see Fig.~\ref{fig:connection-gadget}). If the truth value is $F$, $p_j$ is selected, which (for a realization in $\mathcal{P}^\delta$) forces the choice of a point in a close vicinity of the top-most point of $R_k$. Since there is a stack next to $R_k$, the blue disk in the stack chooses the top-most point, and this eventually forces the choice of $f_{R_\tau}$. The analysis of the other cases are similar.


\subsection{Final construction}

In the following lemma, we put all pieces of the construction together, and we prove that it has polynomial size.


\begin{lemma} \label{lem:fin-cons}
Using the gadgets described above, we can construct in polynomial time a polynomial-size instance of the $\lmcsc$ problem associated to the given planar $\threesat$ formula.
\end{lemma}

\begin{proof}
The dependency graph of the $\threesat$ formula can be embedded so that all variables lie on a horizontal line, all clauses are on either side of them, and each edge connecting a clause to a variable is an orthogonal edge with at most one bend (see Fig.~\ref{fig:pla_emb} for an example)~\cite{KnuthR92}. We take an embedding with these properties which, additionally, is aligned with a grid of resolution 20 (i.e., the distance between any pair of consecutive horizontal or vertical segments is at least 20). In the following lines, we explain how to modify this embedding to produce an instance of the $\lmcsc$ problem. Our construction is illustrated in Fig.~\ref{fig:pla_pic}. 

\begin{figure}[t]
  \centering
  \includegraphics{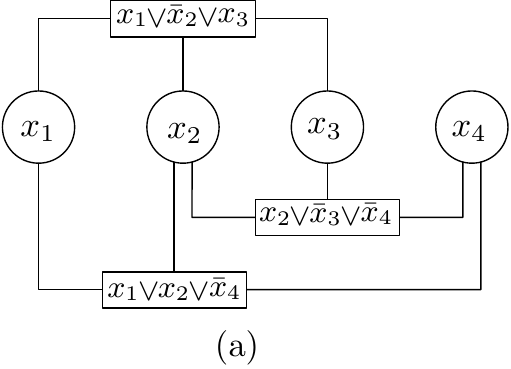}\quad \quad \quad \quad \quad \quad \quad \includegraphics{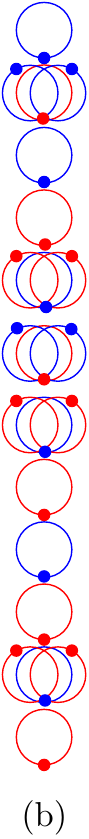}
  \caption{(a) Planar embedding of ${\cal F}=(x_1 \lor x_2 \lor \bar{x}_4)\land (x_1 \lor \bar{x}_2 \lor x_3) \land (x_2 \lor \bar{x}_3 \lor \bar{x}_4)$. (b) A vertical portion of a connection gadget with increased gaps among three central disks, which have been replaced by stacks.}\label{fig:pla_emb}
 \end{figure}
 
First, we replace the variables of the embedding by  gadgets as described above. The size of a variable gadget is chosen so that the upper and lower side of the hexagon are long enough to fit all edges leading to clauses. Since any two such edges are at distance at least 20 and we place them at distance at least 20 from the extremes of the edges of the hexagon, the side of the hexagon is no longer than $20(m+1)$, where $m$ is the number of clauses in the formula. During this step, we might again need to horizontally ``stretch" the initial embedding so that variable gadgets end up at distance at least 20 from each other. 

Next, we replace the clauses of the embedding by clause gadgets. We explain in detail the case where the clause is above the variables; the other case is solved symmetrically (in particular, clause gadgets lying below the variables will be rotated by $180^\circ$). Initially, the clause is incident to three edges, which we call the ``left", ``middle" and ``right" edges. Each of them is assigned to one of the red disks of the clause gadgets, i.e., the connection between the variable gadget and the connection gadget will be done through that disk. In particular, the left edge is assigned to the left-most red disk of the gadget, the middle edge is assigned to the right-most red disk of the gadget, and the right edge is assigned to the top red disk of the gadget (see Fig.~\ref{fig:pla_pic}).

Finally, we replace every edge connecting a clause to a variable gadget by a connection gadget in the following way: The position of the first disk of the connection gadget (that is, the nearest to the clause gadget) is given by the previous assignment of edges to red disks of the clause. The position of the last disk (i.e., the disk of the connection gadget which intersects the variable gadget) is as near as possible to the original location of the edge, while respecting the requirements for each type of connection (which depend on whether the variable appears positive or negative in the clause; refer to Fig.~\ref{fig:connecting}). Since the left and middle edges become connection gadgets with one bend of $120^\circ$ (see Fig.~\ref{fig:pla_pic}), the precise shape of these gadgets (including the position of the bend) is determined by their first and last disks. Notice also that, since the middle edge (which is vertical) is replaced by a chain of shape ``\,\includegraphics[scale=.4]{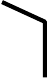}\,'', it might be necessary to move the clause gadget slightly towards the left. Regarding the right edge, it is replaced by a chain of shape ``\,\includegraphics[scale=.4]{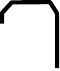}\,'' starting from the top red disk of the clause variable. In this case, there is some freedom to choose the positions of the bends, and we place them so that the final connection gadget is close enough to the original replaced edge. 

After replacing orthogonal edges by connection gadgets with bends of $120^\circ$, some pairs of connection gadgets might have ended up too near (or even cross) and might interfere with each other. This can again be solved by stretching the embedding to separate them.

We finally point out that the straight portions of our connection gadgets are made of chains of unit disks at a certain fixed distance, and thus they can only have certain prescribed lengths of $\frac{9k}{8}+1$ for some integer $k$, which we call \emph{integer} lengths (e.g., for vertical portions of a connection gadget, the distance between the bottom-most and top-most points of the chain can only take values of the form $\frac{9k}{8}+1$, for some integer $k$). Therefore, it might happen that the positions of a clause and variable gadget force the placement of a chain which does not have integer length. In this case, we proceed as follows: Suppose that we need a vertical chain such that the positions and colors of the first and last disk are prescribed, and either the distance $l$ between the bottom-most and top-most points of the chain does not correspond to an integer length, or it does but it has the wrong parity (e.g., when we try to build the chain, the bottom-most disk turns up blue, but we would need it to be red). We take the longest chain of integer length having length smaller than $l$ and having the desired combination of colors of the first and last disk. If this chain has length $\frac{9k}{8}+1$, for some integer $k$, then we have that $l<\frac{9(k+2)}{8}+1$. Thus, the difference in lengths is at most $\frac{9}{4}$. To solve this discrepancy, we select two to four consecutive disks from the middle of the chain and increase the gap between them from $\frac{1}{8}$ to some appropriate value smaller than or equal to $\frac{7}{8}$ (the remaining gaps between consecutive disks of the chain remain of $\frac{1}{8}$). 
By replacing these two to four disks by stacks, we ensure the required propagation among them despite the increased gaps. See Fig.~\ref{fig:pla_emb}b for an example. Since we might increase up to three gaps by up to $\frac{3}{4}$, we can achieve the difference in lengths of at most $\frac{9}{4}$.

Since the final construction has a polynomial number of disks and all the steps in the construction can be performed in polynomial time, the lemma follows.\hfill $\qed$ \end{proof}

\begin{figure}
  \centering
  \includegraphics[angle =90]{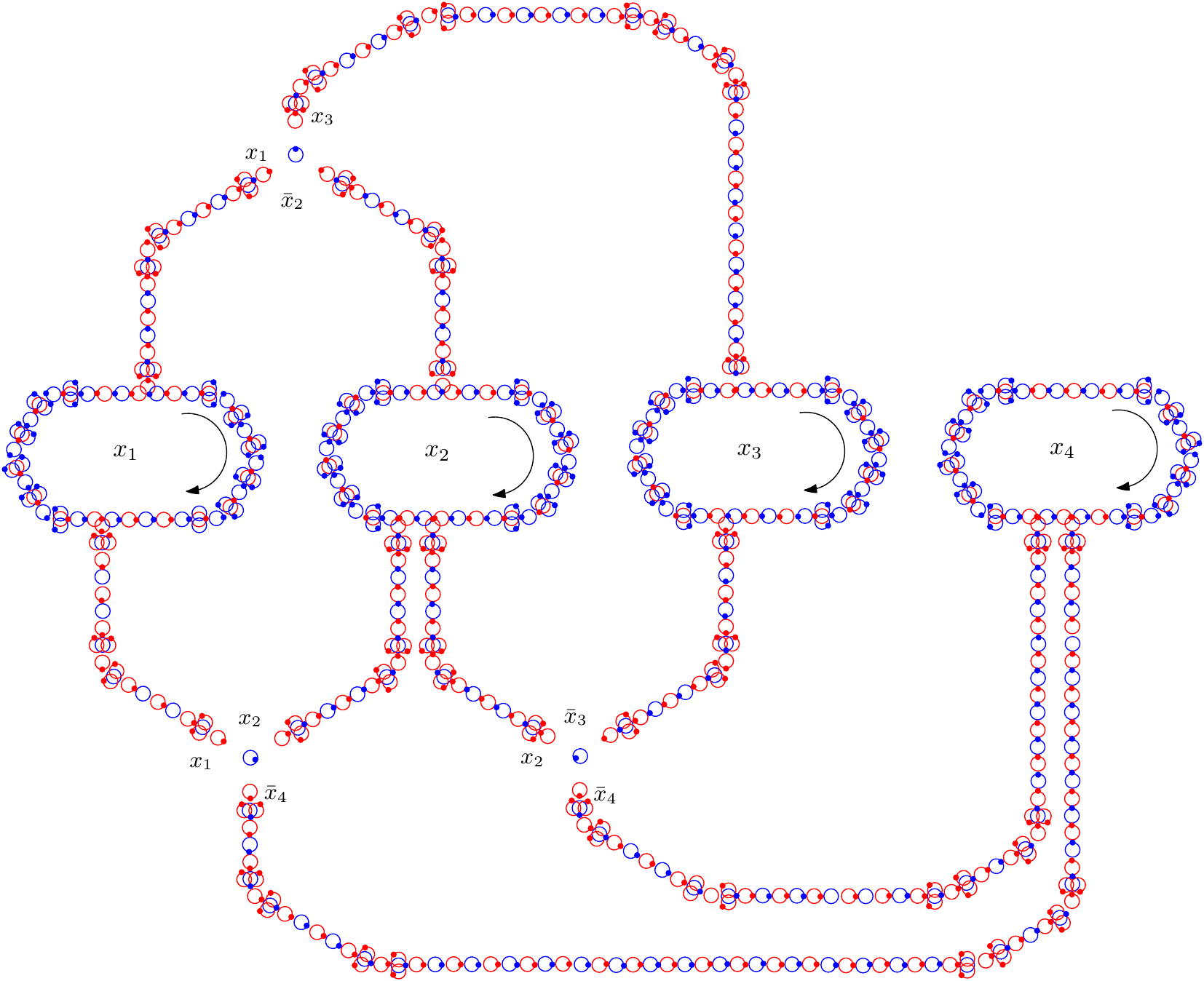}
  \caption{Instance of the $\lmcsc$ problem associated to the formula $\cal F$ in Fig.~\ref{fig:pla_emb}a. For the sake of readability, the figure has been rotated, the proportions among the sizes of the different objects have been altered a little, and the hexagons associated to the variables are not drawn regular. The placement of points  corresponds to the satisfying assignment $x_1=F$, $x_2=T$, $x_3=T$, and $x_4=F$.}\label{fig:pla_pic}
 \end{figure}
 

Given a planar $\threesat$ formula $\cal F$, let us denote by ${\cal R}^{\cal F}$ an associated instance of the $\lmcsc$ problem constructed following the steps in the proof of Lemma~\ref{lem:fin-cons}. We prove that ${\cal R}^{\cal F}$  has the desired property:






 \begin{lemma}\label{lem:nph}
 The planar $\threesat$ formula $\cal F$ has a satisfying assignment if and only if there exists a realization of ${\cal R}^{\cal F}$ in $\mathcal{P}^\delta$.
 \end{lemma}
 
\begin{proof} 
Suppose that $\cal F$ has $n$ variables $x_1,x_2,\ldots,x_n$ and $m$ clauses $C_1,C_2,\ldots,C_m$.

 Let us first assume that $\cal F$ is satisfiable. Thus, every clause contains a literal whose truth value is $T$. We describe a realization of ${\cal R}^{\cal F}$ in $\mathcal{P}^\delta$. 
 
 Let $C_j$ be a clause and let $\ell_i$ be a literal making the clause true. Then we pick $t_{R_i}$ from the associated disk $R_i$ in the clause.
 This choice propagates through the connection gadget and eventually forces one of the realizations for the variable gadget described in Lemma~\ref{lem:var_two_states}: the one choosing the rightmost points of the disks, if $\ell_i$'s truth value is $T$ and it appears in the clause positively; or the one choosing the leftmost points of the disks, if $\ell_i$'s truth value is $F$ and it appears in the clause negatively. Notice that a variable might make true several clauses, but the obtained realizations of the variable gadget are consistent with each other.
 
Let us consider a variable whose realization has been fixed in the paragraph above. If its truth value is $T$ and it appears in a clause negatively, or if its truth value is $F$ and it appears in a clause positively, the placement described above forces the realizations of the corresponding connection gadget and associated disk in the clause gadget. 

For the variables whose realization has not been fixed yet (if any), we pick, say, the realization choosing the leftmost points of the disks. For their connection gadgets to the clauses, we pick the realization choosing the top-most (respectively, bottom-most) points of the disks, if the clause is above (respectively, below) the variable. The described realizations belong to $\mathcal{P}^\delta$.

After following the steps above, for each clause gadget at least one of the three red disks has its representative at the position $t_{R_\tau}$. In consequence, it is possible to find a realization of the blue disk such that the realization of the clause gadget is in $\mathcal{P}^\delta$. 

Notice that we have described realizations for all gadgets in the construction. Clearly, the global realization is in $\mathcal{P}^\delta$.
 
We next prove the second implication.
Suppose that there exists a realization $R$ of ${\cal R}^{\cal F}$ in $\mathcal{P}^\delta$. 
 Let $C_j$ be a clause. Since the realization of the clause is in $\mathcal{P}^\delta$, at least one of the red disks $R_i$ of the clause does not have the representative at the position $f_{R_i}$. If the corresponding literal appears positive in the clause, we set the associated variable to $T$. Otherwise, we set it to $F$. In this way, we ensure that every clause is true, but we need to argue that we did not assign $T$ and $F$ simultaneously to the same variable due to two distinct clauses.
 
 Since in the connection gadget there is a stack of disks next to $R_i$, the fact that $R_i$ does not have the representative at the position $f_{R_i}$ forces the choice of the representative of the blue disk in the stack. The choice propagates to the connection gadget and eventually to the variable gadget: If the literal associated to $R_i$ appears positive in $C_j$, the variable has been set to $T$ and the realization in the variable gadget is the one choosing the rightmost points of the disks.  If the literal appears negative, the variable has been set to $F$ and the realization in the variable gadget chooses the leftmost points of the disks. Since in $R$ the realization of the variable gadget is either the rightmost or the leftmost, the variable has either been set to $T$ or to $F$. \hfill $\qed$ \end{proof}

Lemmas~\ref{lem:fin-cons} and~\ref{lem:nph} imply the following:

\begin{theorem}
The problem of finding a largest minimum color spanning circle of $\cal R$ is $\nph$.
\end{theorem}

\begin{remark}
Given a yes-instance, we can verify in $O(nk \log n)$ time whether the given realization is correct and the radius of its $\mcsc$ is at least $\delta$. Therefore, the decision version of the problem is  $\npc$.
\end{remark}

\section{Approximation algorithms}

Given that the $\lmcsc$ problem is $\nph$, in this section we turn our attention to approximation algorithms. 

Let $\tilde{r}_{opt}$ denote the radius of a largest possible minimum color spanning circle of ${\cal R}$. We first prove bounds on $\tilde{r}_{opt}$. 

\begin{figure}[t]
 \centering
 \includegraphics{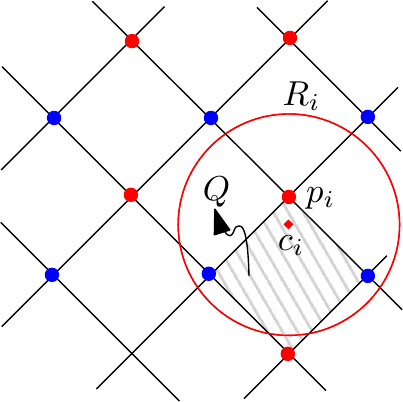}
 \caption{The tilted grid in the proof of Lemma~\ref{lem:low-bd-lmcsc}. We choose $p_i$ as the red corner of $Q$ contained in $R_i$.}\label{fig:lower-bound-radius}
\end{figure}

\begin{lemma} \label{lem:low-bd-lmcsc}
$\tilde{r}_{opt}\geq 1/4$.
\end{lemma}

\begin{proof}
First, we show that it is enough to prove the result for the case where $k=2$: Suppose that $k>2$. We select two of the $k$ colors, say red and blue, and the set $\cal R'\subseteq \cal R$ of red and blue disks. Suppose that the radius of any $\lmcsc$ of $\cal R'$ is greater than or equal to $1/4$. Then there exists a realization of $\cal R'$ whose $\mcsc$ has radius greater than or equal to $1/4$. We complete this realization of $\cal R'$ to a realization of $\cal R$ by choosing any point in each disk of the remaining $k-2$ colors. Since any $\mcsc$ of this realization contains at least one red and one blue point, its radius is greater than or equal to $1/4$.

 Therefore, we only show that the bound holds when $k=2$ by providing a realization $P$ whose $\mcsc$ achieves the bound. Consider a regular square grid rotated by $\pi/4$ such that the side of every cell of the grid has length $1/2$. We color the corners of the cells in red or blue in such a way that all corners lying in some vertical line are colored red, all corners lying in the next vertical line are colored blue, and so on (see Fig.~\ref{fig:lower-bound-radius}). Now let $R_i\in \cal R$ have red color, and let $Q$ be the cell of the grid containing the center of $R_i$ (if the center of $R_i$ lies on an edge or vertex of the grid, we assign it to any of the adjacent cells). Notice that at least one of the two red corners of $Q$ lies inside $R_i$. We choose such a corner as $p_i\in P$. Similarly, for every $R_j\in \cal R$ of blue color, $P$ contains one of the blue corners of a cell containing the center of $R_j$. We obtain that $P$ is a subset of the grid corners. Since the distance between any pair of red and blue corners is at least $1/2$, the radius of any  $\mcsc$ is at least $1/4$. \hfill $\qed$ 
 \end{proof}

Next, we establish a relation between $\tilde{r}_{opt}$ and $r_c$. We recall that $r_c$ is the radius of a $\mcsc$ (denoted $C_c$) of the colored set $\cal C$ containing the center points of the disks in $\cal R$.

 \begin{lemma}\label{lem:apx}
$\tilde{r}_{opt} \leq r_c + \frac{1}{2}$.
\end{lemma}
\begin{proof}
Without loss of generality we can assume that $C_c$ contains the centers of the disks in ${\cal R'}=\{R_1,R_2,\ldots,R_k\}$. A circle of radius $r_c + \frac{1}{2}$, concentric with $C_c$,  contains the $k$ disks in ${\cal R'}$.
Therefore, for any placement of points inside disks of ${\cal R'}$, we can always get a color spanning circle with radius at most $r_c + \frac{1}{2}$. Thus, $\tilde{r}_{opt}\leq r_c + \frac{1}{2}$. \hfill $\qed$ 
\end{proof}

We use the bounds above to design a simple approximation algorithm, presented in Algorithm~\ref{algo:3-approx-alg}.
Let $P^g$ denote the realization of $\cal R$ described in the proof of Lemma~\ref{lem:low-bd-lmcsc}.\\ 

\begin{tcolorbox}[rightrule=2pt,toprule=0.5pt, bottomrule=2pt,leftrule=0.5pt,arc=8pt]
\captionof{algorithm}{$\frac{1}{3}$-factor approximation algorithm for the $\lmcsc$ problem} \label{algo:3-approx-alg}
\hrulefill\vskip -2mm
\begin{algorithmic}[1]
\Require {\it A set ${\cal R}$ of $n$ unit disks}
\Ensure {\it A $\mcsc$ of a realization of $\cal R$ with radius at least $\tilde{r}_{opt}/3$}
\State compute $C_c$;
\If {$r_c\geq 1/4$}
    \State \Return $C_c$;
\Else
    \State \Return a $\mcsc$ of $P^g$;
\EndIf
\end{algorithmic}
\end{tcolorbox} 

It only remains to prove that Algorithm~\ref{algo:3-approx-alg} indeed gives a $\frac{1}{3}$-factor approximation:

\begin{theorem}\label{thm:approx-alg}
A $\frac{1}{3}$-factor approximation for the $\lmcsc$ problem can be computed in $O(nk\log n)$ time. If no two distinct colored unit disks of~$\cal R$ intersect, the approximation factor becomes $\frac{1}{2}$.
\end{theorem}
  \begin{proof}
If $r_c\geq 1/4$, then, by Lemma~\ref{lem:apx}, $\tilde{r}_{opt}\leq 3r_c$ and $C_c$ gives a $\frac{1}{3}$-approximation for the problem. If $r_c< 1/4$, then, by Lemma~\ref{lem:apx}, $\tilde{r}_{opt}<3/4$. By Lemma~\ref{lem:low-bd-lmcsc}, any $\mcsc$ of $P^g$ has radius at least $1/4$, so such a circle gives a $\frac{1}{3}$-approximation for the problem.

Computing $P^g$ takes $O(n)$ time, and computing a $\mcsc$ of $P^g$ or the set of disk centers can be done in $O(nk\log n)$ time.

Finally, if no two distinct colored disks intersect, $r_c\geq \frac{1}{2}$. This, combined with Lemma~\ref{lem:apx}, gives $r_c\geq \frac{\tilde{r}_{opt}}{2}$. \hfill $\qed$ 
\end{proof}

 


Since Lemma~\ref{lem:low-bd-lmcsc} is one of the key ingredients of our approximation algorithm, we conclude this section by showing that it is the best possible.

\begin{figure}[t]
 \centering
\includegraphics{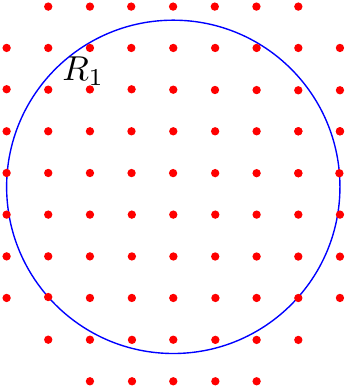}
 \caption{ Construction in the proof of Remark~\ref{rem:upp-bd-lmcsc}. For the sake of clarity, the red disks centered at the grid points are not displayed.}\label{fig:lower-bound-radiusr}
\end{figure}

\begin{remark} \label{rem:upp-bd-lmcsc}
Given any $\varepsilon>0$, there exists a set ${\cal R}$ of unit disks for which $\tilde{r}_{opt}<1/4+\varepsilon$.
\end{remark}

\begin{proof}
We start with a blue disk $R_1$. We overlay a square grid of red points over the area covered by $R_1$ enlarged a little around the boundary (see Fig.~\ref{fig:lower-bound-radiusr}). The length of the sides of the grid cells is set to $2\sqrt{2}\varepsilon$. Then, for every red point of the grid, we place a red disk centered at that point. Finally, we place the remaining blue disks far away from the construction.

Let $p_1$ be the representative for $R_1$ in a realization $P$ giving a minimum color spanning circle of radius $\tilde{r}_{opt}$. Since $p_1$ lies inside some cell of the grid, there exists some red grid point $q$ such that $d(p_1,q)\leq 2\varepsilon$. Let $R_j$ be the red disk centered at $q$, and $p_j$ be its representative in the  realization $P$. Then, $d(p_1,p_j)\leq d(p_1,q)+d(q,p_j)\leq 2\varepsilon+1/2$. Since $p_1,p_j$ is a blue-red pair at distance at most $1/2+2\varepsilon$, the radius of any minimum color spanning circle is at most $1/4+\varepsilon$. \hfill $\qed$ 
\end{proof}

\section{Open problems}

Naturally, the main open problems are related to the $\lmcsc$ problem and the existence of better approximation algorithms for it. So far we have not succeeded in finding a PTAS for this problem, and we suspect that the problem might be $\apx$. If this is the case, it would also be interesting to improve the approximation factor of our approximation algorithm.

\paragraph{Acknowledgements}
The authors would like to thank Irina Kostitsyna for key discussions on the hardness reduction and Hans Raj Tiwary for the proof of Lemma~\ref{lem:low-bd-lmcsc}. A.A., R.J., V.K., and M. S. were supported by the Czech Science Foundation, grant number GJ19-06792Y, and by institutional support RVO: 67985807.  M.L. was partially supported by the Netherlands Organization for Scientific Research (NWO) under project no. 614.001.504. This project has received funding from the European Union's Horizon 2020 research and innovation programme under the Marie Sk\l{}odowska-Curie grant agreement No 734922.

\bibliographystyle{splncs04}
\bibliography{mcsc_full.bib}
\end{document}